%% file: ufa_main.tex
\documentclass[a4paper,UKenglish,cleveref,nameinlink,thm-restate]{lipics-v2021}

\pdfoutput=1 
\hideLIPIcs  

\usepackage[dvipsnames]{xcolor}

\usepackage{amsmath,amsthm,amsfonts,amssymb,mathtools,stmaryrd}

\usepackage[framemethod=tikz]{mdframed}

\usepackage{enumitem} 

\title{Lower Bounds for Unambiguous Automata\\ via Communication Complexity} 
\titlerunning{Lower Bounds for Unambiguous Automata via Communication Complexity}

\author{Mika G\"o\"os}{EPFL, Switzerland
}{mika.goos@epfl.ch}{}{}
\author{Stefan Kiefer}{University of Oxford, UK
}{stefan.kiefer@cs.ox.ac.uk}{}{}
\author{Weiqiang Yuan}{EPFL, Switzerland
}{weiqiang.yuan@epfl.ch}{}{}

\authorrunning{M. G\"o\"os, S. Kiefer, W. Yuan} 
\Copyright{Mika G\"o\"os, Stefan Kiefer, and Weiqiang Yuan}



\category{} 

\relatedversion{} 

\nolinenumbers 

\EventEditors{John Q. Open and Joan R. Access}
\EventNoEds{2}
\EventLongTitle{42nd Conference on Very Important Topics (CVIT 2016)}
\EventShortTitle{CVIT 2016}
\EventAcronym{CVIT}
\EventYear{2016}
\EventDate{December 24--27, 2016}
\EventLocation{Little Whinging, United Kingdom}
\EventLogo{}
\SeriesVolume{42}
\ArticleNo{23}

\input{environment}

\begin{document}

\maketitle

\begin{abstract}
We use results from communication complexity, both new and old ones, to prove lower bounds for unambiguous finite automata (UFAs). We show three results.
\begin{enumerate}[topsep=3pt,noitemsep]
\item \emph{Complement:}
There is a language $L$ recognised by an $n$-state UFA such that the complement language $\overline{L}$ requires NFAs with $n^{\tilde{\Omega}(\log n)}$ states. This improves on a lower bound by Raskin.
\item \emph{Union:}
There are languages $L_1$, $L_2$ recognised by $n$-state UFAs such that the union $L_1\cup L_2$ requires UFAs with $n^{\tilde{\Omega}(\log n)}$ states.
\item \emph{Separation:}
There is a language $L$ such that both $L$ and $\overline{L}$ are recognised by $n$-state NFAs but such that $L$ requires UFAs with $n^{\Omega(\log n)}$ states. This refutes a conjecture by Colcombet.
\end{enumerate}
\end{abstract}

\input{parts/introduction}

\input{parts/preliminaries}

\input{parts/complementation}

\input{parts/union}

\input{parts/separation}

\input{parts/applications}


\end{document}

%% file: environment.tex
\hypersetup{
    colorlinks=true,
	citecolor=MidnightBlue,
    urlcolor=MidnightBlue,
	linkcolor=MidnightBlue
}

\DeclareUrlCommand{\Doi}{\urlstyle{sf}}
\renewcommand{\path}[1]{\small\Doi{#1}}
\renewcommand{\url}[1]{\href{#1}{\small\Doi{#1}}}

\crefname{equation}{}{}
\crefname{proposition}{Proposition}{Propositions}
\crefname{theorem}{Theorem}{Theorems}
\crefname{corollary}{Corollary}{Corollaries}
\crefname{algorithm}{Algorithm}{Algorithms}
\crefname{section}{Section}{Sections}
\crefname{lemma}{Lemma}{Lemmas}

\newcommand{\A}{\mathcal{A}}

\newcommand{\C}{\mathrm{C}}
\newcommand{\Cov}{\mathrm{Cov}}

\newcommand{\Disj}{\textsc{Disj}}
\newcommand{\enc}[1]{\langle #1 \rangle}
\newcommand{\N}{\mathbb{N}}
\newcommand{\Non}{\mathrm{N}}
\newcommand{\Par}{\mathrm{Par}}
\newcommand{\tO}{\tilde{O}}
\newcommand{\tOmega}{\tilde{\Omega}}
\newcommand{\UC}{\mathrm{UC}}
\newcommand{\Una}{\mathrm{U}}

\newcommand{\poly}{\mathrm{poly}}
\newcommand{\rank}{\mathrm{rk}}
\newcommand{\nrank}{\rank^+}
\newcommand{\anrank}[1]{\rank_{#1}^+}
\newcommand{\ndeg}{\deg^+}
\newcommand{\andeg}[1]{\deg^+_{#1}}
\newcommand{\ORf}{f^{\vee}}

\newcommand{\bN}{\mathbb{N}}
\newcommand{\bR}{\mathbb{R}}

\newcommand{\cC}{\mathcal{C}}

\newcommand{\bool}{\{0,1\}}

\mdfsetup{
	roundcorner=4pt,
	nobreak=true,
	linewidth=.5,
	innerleftmargin=\parindent,
	innerrightmargin=15pt,
	innertopmargin=0pt,
	innerbottommargin=11pt,
	skipabove=12,skipbelow=5pt
}

%% file: parts/introduction.tex
\section{Introduction} \label{sec:intro}

Given two finite automata recognising languages $L_1, L_2 \subseteq \Sigma^*$ a basic question is to determine the \emph{state complexity} of various language operations. How many states are needed in an automaton that recognises the union $L_1 \cup L_2$? How about the intersection $L_1 \cap L_2$? The complement~$\overline{L}_1\coloneqq\Sigma^* \setminus L_1$? The answer depends on the type of automaton considered, such as deterministic~(DFA), nondeterministic (NFA), or unambiguous (UFA). Recall that a UFA is an NFA that has at most one accepting computation on any input. 

State complexities have been extensively studied for various types of automata and language operations; see, e.g., \cite{GaoMRY16,JirasekJS18} and their references, or the excellent compendium on Wikipedia~\cite{Wiki}.
For example, complementing an NFA with $n$~states may require $2^n$ states~\cite{Birget1993}, even for automata with binary alphabet~\cite{Jiraskova05}.
Surprisingly, several extremely basic questions about~UFAs remain open. For example, it was shown only in~2018 by Raskin~\cite{Raskin18} that the state complexity for UFA complementation is not polynomial: for any $n \in \N$ there exists a language $L$ recognised by an $n$-state UFA such that any UFA (or even NFA) that recognises~$\overline{L}$ has at least $n^{(\log \log \log n)^{\Omega(1)}}$ states. This superpolynomial blowup refuted a conjecture that it may be possible to complement UFAs with a polynomial blowup~\cite{Colcombet15}.

In this paper, as our main results, we prove three new blowup theorems.
\begin{restatable}[Complement]{theorem}{complement} \label{thm:complement}
For every $n \in \N$ there is a language $L\subseteq\{0,1\}^*$ recognised by an $n$-state UFA such that any NFA that recognises $\overline{L}$ requires $n^{\tOmega(\log n)}$ states.
\end{restatable}%
\vspace{-3mm}
\begin{restatable}[Union]{theorem}{union} \label{thm:union}
For every $n \in \N$ there are languages $L_1,L_2\subseteq\{0,1\}^*$ recognised by $n$-state UFAs such that any UFA that recognises $L_1\cup L_2$ requires $n^{\tOmega(\log n)}$ states.
\end{restatable}%
\vspace{-3mm}
\begin{restatable}[Separation]{theorem}{separation} \label{thm:separation}
For every $n\in \N$ there is a language $L\subseteq\{0,1\}^*$ such that both $L$ and $\overline{L}$ are recognised by $n$-state NFAs but any UFA that recognises $L$ requires $n^{\Omega(\log n)}$ states.
\end{restatable}

\subparagraph{Discussion of main results.}
{\bf \cref{thm:complement}} upgrades Raskin's slightly-superpolynomial bound into a quasipolynomial bound $n^{\tOmega(\log n)}$. (Here we use the notation~$\tOmega(m)$ to suppress $\poly(\log m)$ factors.) However, we note that Raskin's language is unary, $|\Sigma|=1$, while ours is binary, $|\Sigma|=2$, and hence the two results are incomparable in this sense. As for positive results, it is known that the trivial $2^n$ upper bound for UFA complementation can be improved: the complement of any $n$-state UFA can be recognised by a UFA with at most~$\poly(n)\cdot2^{n/2}$ states~\cite{JirasekJS18,IK21}. Closing the exponential gap here between the lower and upper bounds remains a tantalising open problem. It was highlighted as one of the foremost challenges in the recent Dagstuhl workshop \emph{Unambiguity in Automata Theory}~\cite{Colcombet2021}.

\bigskip\noindent
{\bf \cref{thm:union}} establishes the first superpolynomial lower bound for the union operation. Letting $\sqcup$ denote disjoint union, observe that
\begin{equation}\label{eq:rewrite}
L_1\cup L_2 ~=~ L_1 \sqcup (L_2\cap \overline{L}_1).
\end{equation}
Since disjoint union and intersection are polynomial for UFAs, it follows from \cref{eq:rewrite} and \cref{thm:union} that the same~$n^{\tOmega(\log n)}$ lower bound holds for complementing UFAs. However, we stress that \cref{thm:complement} has a stronger conclusion than this, since it proves a lower bound against NFAs, not just UFAs. The observation~\cref{eq:rewrite} also yields the upper bound $\poly(n)\cdot 2^{n/2}$ by using the complement construction from~\cite{JirasekJS18,IK21}. 

\bigskip\noindent
{\bf \cref{thm:separation}} refutes a conjecture by Colcombet~\cite[Conjecture~2]{Colcombet15}. Indeed, he conjectured that for any pair of NFAs recognising languages $L_1$, $L_2$ such that $L_1\cap L_2=\emptyset$, there is a polynomial-sized UFA that recognises some $L$ that \emph{separates} $L_1$ and $L_2$ in the sense that~$L_1\subseteq L$ and $L\cap L_2=\emptyset$. \cref{thm:separation} refutes this even in the special case $L_1=\overline{L}_2$. Related separability questions are classical in formal language theory and have attracted renewed attention; see, e.g, \cite{CzerwinskiL19} and the references therein. Separating automata have also been used recently to elegantly describe quasipolynomial time algorithms for solving parity games in an automata theoretic framework; see \cite[Chapter~3]{automata-toolbox} and~\cite{CzerwinskiD19}.

\subsection{Technique: Communication complexity}

Our three main theorems rely on results---both new and old---in communication complexity; see~\cite{KushilevitzNisan,Rao2020} for the standard textbooks. In communication complexity, one studies functions of the form~$F\colon \{0,1\}^n \times \{0,1\}^n \to \{0,1\}$ that determine the following two-party communication problem: Alice holds $x\in \{0,1\}^n$, Bob holds $y\in \{0,1\}^n$, and their goal is to output $F(x,y)$ while communicating as few bits as possible between them. Communication complexity is a classical tool to prove lower bounds for automata. Indeed, it is well known that if the language~$\{xy:F(x,y)=1\}$ is recognised by a small DFA (resp.\ NFA, UFA) then $F$ admits an efficient deterministic (resp.\ nondeterministic, unambiguous) protocol. We revisit this connection in light of recent developments in communication complexity.

\bigskip\noindent
{\bf \cref{thm:complement}} is a \emph{\bf\itshape relatively straightforward consequence} of a recent result of Balodis et al.~\cite{Balodis2021FOCS}. They exhibited a two-party function whose co-nondeterministic communication complexity is nearly quadratic in its unambiguous complexity (which matches an upper bound due to Yannakakis~\cite{Yannakakis1991}). We translate this separation into the language of automata theory, virtually in a black-box fashion.

\bigskip\noindent
{\bf \cref{thm:union}}, by contrast, is our \emph{\bf\itshape main technical contribution}. We will show that it follows from the following analogous communication result, which we prove in this paper.
\begin{theorem}\label{thm:or}
For every $m\in\N$ there exists a function $F(x,y)$ with unambiguous communication complexity at most $m$ such that the logical-or of two copies of $F$, namely, $F^{\lor}(xx',yy')\coloneqq F(x,y)\lor F(x',y')$, has unambiguous communication complexity $\tOmega(m^2)$.
\end{theorem}
This is a new result in communication complexity; the unambiguous complexity of $F^\lor$ has not been studied previously. We prove \cref{thm:or} using the popular \emph{query-to-communication lifting} technique that has been wildly successful in the past decade to prove communication lower bounds (including in~\cite{Balodis2021FOCS}). In this technique, one starts by proving a lower bound on the query (aka decision tree) complexity of a boolean function $f\colon\{0,1\}^n\to\{0,1\}$. A lifting theorem (e.g.,~\cite{GLMWZ16}) then transforms $f$ into an analogous communication problem $F$ in such a way that the communication complexity of $F$ is characterised by the query complexity of~$f$. This reduces the task of proving communication lower bounds into the much easier task of proving query lower bounds.

Interestingly, our proof of \cref{thm:or} formalises a kind of \emph{converse} to the observation~\cref{eq:rewrite} above (saying that union can be computed via a complement). Namely, we show that unambiguously computing the union necessarily requires computing a complement, and therefore we can rely on an existing query lower bound for complementation~\cite{GJPW18}.

\bigskip\noindent
{\bf \cref{thm:separation}}, finally, is a straightforward consequence of a classical quadratic separation between two-sided nondeterministic communication complexity and unambiguous communication complexity due to Razborov~\cite{Razborov1990}.

\subsection{Bonus result: Approximate nonnegative rank}

Along the way to \cref{thm:or} we inadvertently stumbled upon another separation result that addresses a question raised by Kol et al.~\cite{KMSY14}. They studied the \emph{$\epsilon$-approximate nonnegative rank} $\anrank{\epsilon}(M)$ of a nonnegative matrix $M\in \bR^{n\times n}$. Here, $\anrank{\epsilon}(M)$ is defined as the least nonnegative rank $\nrank(N)$ of a matrix $N\in \bR^{n\times n}$ such that $|M_{ij}-N_{ij}|\leq\epsilon$ for all $i,j$; see \cref{sec:union} for precise definitions. In particular, Kol et al.~\cite{KMSY14} asked whether for all error parameters $0<\epsilon<\delta<1/2$ and boolean matrices $M\in\bool^{n\times n}$ we have the polynomial relationship $\anrank{\epsilon}(M) \leq O(\anrank{\delta}(M)^C)$ where $C=C(\epsilon,\delta)$ is a constant. In short, does approximate nonnegative rank admit efficient error reduction? (It is known that the more usual notion, \emph{approximate rank}, does~\cite{Alon2003}.) We provide the following negative answer.
\begin{restatable}[No efficient error reduction]{theorem}{error} \label{thm:error}
For every $m\in\bN$ there exists a boolean matrix $M$ with $\anrank{1/4}(M)\leq m$ but such that \smash{$\anrank{10^{-5}}(M)\geq m^{\tOmega(\log m)}$}.
\end{restatable}%
Previously, a negative answer was known only for \emph{partial} boolean matrices $M\in\{0,1,*\}^{n\times n}$ that allow ``don't care'' entries $M_{ij}=*$~\cite{GLMWZ16}. Our \cref{thm:error} still leaves open the possibility (also raised by~\cite{KMSY14}) that, for a total boolean matrix $M$, we can bound $\anrank{\epsilon}(M)$ as a polynomial function of $\anrank{\delta}(M)+\anrank{\delta}(\overline{M})$ where $\overline{M}$ is the boolean complement.

\subsection{Open problems}

Our quasipolynomial lower bounds for automata are not known to be tight; in all cases the best known upper bounds are exponential. Curiously enough, the analogous communication results are tight for communication protocols. This suggests two opportunities.
\begin{itemize}[topsep=3pt]
\item Can other techniques from communication complexity improve the lower bounds further? Perhaps via multi-party communication complexity?
\item Can techniques for proving upper bounds on communication complexity be adapted to prove upper bounds on the size of automata?
\end{itemize}

%% file: parts/preliminaries.tex
\subsection{Definitions of automata} \label{sec:prel}

An \emph{NFA} is a quintuple $\A = (Q,\Sigma,\delta,I,F)$, where $Q$ is the finite set of states, $\Sigma$ is the finite alphabet, $\delta \subseteq Q \times \Sigma \times Q$ is the transition relation, $I \subseteq Q$ is the set of initial states, and $F \subseteq Q$ is the set of accepting states.
We write $q \xrightarrow{a} r$ to denote that $(q,a,r) \in \delta$.
A finite sequence $q_0 \xrightarrow{a_1} q_1 \xrightarrow{a_2} \cdots \xrightarrow{a_n} q_n$ is called a \emph{run}; it can be summarized as $q_0 \xrightarrow{a_1 \cdots a_n} q_n$.
The NFA~$\A$ \emph{recognizes} the language $L(\A) := \{w \in \Sigma^* \mid \exists\, q_0 \in I \,.\, \exists\, f \in F \,.\, q_0 \xrightarrow{w} f\}$.
The NFA~$\A$ is a \emph{DFA} if $|I| = 1$ and for every $q \in Q$ and $a \in \Sigma$ there is exactly one $q'$ with $q \xrightarrow{a} q'$.
The NFA~$\A$ is a \emph{UFA} if for every word $w = a_1 \cdots a_n \in \Sigma^*$ there is at most one \emph{accepting} run for~$w$, i.e., a run $q_0 \xrightarrow{a_1}  q_1 \xrightarrow{a_2} \cdots \xrightarrow{a_n} q_n$ with $q_0 \in I$ and $q_n \in F$.
Any DFA is a UFA.


%% file: parts/complementation.tex
\section{UFA Complementation} \label{sec:complement}


In this section we prove \cref{thm:complement}.
\complement*
The proof uses concepts from communication complexity, in particular a recent result from~\cite{Balodis2021FOCS} and a nondeterministic lifting theorem from~\cite{GLMWZ16}. We start by recalling these tools.

\subsection{DNFs and nondeterministic protocols} \label{sub:com-com}
\subparagraph{Unambiguous DNFs.}
Let $D = C_1 \lor \cdots \lor C_m$ be an $n$-variate boolean formula in disjunctive normal form (DNF).
DNF~$D$ has \emph{width}~$k$ if every $C_i$ is a conjunction of at most $k$ literals.
We call such~$D$ a \emph{$k$-DNF}.
For conjunctive normal form (CNF) formulas the width and $k$-CNFs are defined analogously.
DNF~$D$ is said to be \emph{unambiguous} if for every input $x \in \{0,1\}^n$ at most one of the conjunctions~$C_i$ evaluates to true, $C_i( x ) = 1$.
For any boolean function $f\colon \{0,1\}^n \to \{0,1\}$ define 
\begin{itemize}
\item $\C_1(f)$ as the least $k$ such that $f$ can be written as a $k$-DNF;
\item $\C_0(f)$ as the least $k$ such that $f$ can be written as a $k$-CNF;
\item $\UC_1(f)$ as the least $k$ such that $f$ can be written as an unambiguous $k$-DNF.
\end{itemize}
Note that $\C_0(f) = \C_1(\neg f)$.
The following recent result separates two of these measures.
\begin{theorem}[{\cite[Theorem~1]{Balodis2021FOCS}}] \label{thm:Puzzle-I}
For every~$k\in\bN$ there exists a function $f\colon\{0,1\}^n \to \{0,1\}$ where $n\leq \poly(k)$ and such that $\UC_1(f) \leq k$ and $\C_0(f) \geq \tOmega(k^2)$.\qed
\end{theorem}
In words, for every~$k$ there is an unambiguous $k$-DNF such that any equivalent CNF requires width~$\tOmega(k^2)$.
The bound is almost tight, as every unambiguous $k$-DNF has an equivalent $k^2$-CNF; see~\cite[Section~3]{Goos15}.

\subparagraph{Nondeterministic protocols and rectangle covers.}
Next we recall standard notions from two-party communication complexity; see \cite{KushilevitzNisan,Rao2020} for textbooks.
Consider a two-party function $F\colon X \times Y \to \{0,1\}$.
A set $A \times B \subseteq X \times Y$ (with $A \subseteq X$ and $B \subseteq Y$) is called a \emph{rectangle}.
Rectangles $R_1, \ldots, R_k$ \emph{cover} a set $S \subseteq X \times Y$ if $\bigcup_i R_i = S$.
For~$b \in \{0,1\}$, the \emph{cover number} $\Cov_b(F)$ is the least number of rectangles that cover $F^{-1}(b)$.
The \emph{nondeterministic (resp., co-nondeterministic) communication complexity} of~$F$ is defined as $\Non_1(F) \coloneqq \log_2 \Cov_1(F)$ (resp., $\Non_0(F) := \log_2 \Cov_0(F)$).
Note that $\Non_0(F) = \Non_1(\neg F)$.
The nondeterministic communication complexity can be interpreted as the number of bits that two parties (Alice and Bob), holding inputs $x \in X$ and $y \in Y$, respectively, need to communicate in a nondeterministic (i.e., based on guessing and checking) protocol in order to establish that~$F(x,y) = 1$; see~\cite[Chapter~2]{KushilevitzNisan} for details.

\subparagraph{Nondeterministic lifting.}
Next we formulate a \emph{lifting theorem}, which allows us to transfer lower bounds on the DNF width of an $n$-bit boolean function $f$ to the nondeterministic communication complexity of a related two-party function $F$. We first choose a small two-party function $g\colon \{0,1\}^b \times \{0,1\}^b \to \{0,1\}$, often called a \emph{gadget}. Then we compose $f$ with~$g$ to construct the function $F\coloneqq f\circ g^n$ that maps $\{0,1\}^{b n} \times \{0,1\}^{b n} \to \{0,1\}$ where Alice gets as input $x\in \{0,1\}^{b n}$, Bob gets as input $y\in \{0,1\}^{b n}$, and their goal is to compute
\[
F(x,y) ~\coloneqq~ f(g(x_1,y_1), \ldots, g(x_n,y_n)) \qquad \text{where } x_i,y_j \in \{0,1\}^b.
\]
The following is a nondeterministic lifting theorem~\cite{GLMWZ16,Goos15}.
\begin{theorem}[{\cite[Theorem~4]{Goos15}}] \label{thm:lifting}
For any $n \in \N$ there is a gadget $g\colon \{0,1\}^b \times \{0,1\}^b \to \{0,1\}$ with $b = \Theta(\log n)$ such that for any function $f\colon \{0,1\}^n \to \{0,1\}$ we have, for $F\coloneqq f\circ g^n$,
\[
\Non_0(F) ~=~ \Omega(\C_0(f) \cdot b)
\]
(and thus also $\Non_1(F) = \Omega(\C_1(f) \cdot b)$).\qed
\end{theorem}

\subparagraph{Protocols can simulate automata.}
Finally, we need a simple folklore connection between automata and protocols. To formalise this, we tacitly identify a function $F\colon \{0,1\}^{m_1} \times \{0,1\}^{m_2} \to \{0,1\}$ with the language $F^{-1}(1) ~=~ \{x y \in \{0,1\}^{m_1+m_2} \mid F(x,y) = 1\}$.
\begin{lemma} \label{lem:NFA-CC}
If a two-party function $F\colon \{0,1\}^{m_1} \times \{0,1\}^{m_2} \to \{0,1\}$ admits an NFA with $s$~states, then $\Cov_1(F) \le s$ (that is, $\Non_1(F)\leq\log s$).
\end{lemma}
\begin{proof}
Let $\A = (Q,\Sigma,\delta,I,F)$ be an NFA with $L(\A) = \{x y \in \{0,1\}^{m_1+m_2} \mid F(x,y) = 1\}$.
We show that $F^{-1}(1)$ is covered by at most $|Q|$ rectangles.
Indeed, $F^{-1}(1)$ equals
\[
 \bigcup_{q \in Q} (\{x \in \{0,1\}^{m_1} \mid \exists\, q_0 \in I \,.\, q_0 \xrightarrow{x} q\}) \times 
                   (\{y \in \{0,1\}^{m_2} \mid \exists\, f \in F \,.\, q \xrightarrow{y} f\}) \,.
\]
(Alternatively, in terms of a nondeterministic protocol, the first party, holding $x \in \{0,1\}^{m_1}$, produces a run for~$x$ from an initial state to a state~$q$ and then sends the name of~$q$, which takes $\log_2 |Q|$ bits, to the other party.
The other party then produces a run for~$y$ from $q$ to an accepting state.)
\end{proof}

\subsection{Proof of \texorpdfstring{\cref{thm:complement}}{\cref{thm:complement}}} \label{sec:proof-1}

For $k \in \N$, let $f\colon\{0,1\}^n \to \{0,1\}$ be the function from \cref{thm:Puzzle-I}. That is, $f$ has an unambiguous $k$-DNF with $k = n^{\Omega(1)}$ (hence, $\log n = O(\log k)$) and $\C_0(f) = \tOmega(k^2)$.
Let $g\colon\{0,1\}^b \times \{0,1\}^b \to \{0,1\}$ with $b = \Theta(\log n)$ and $F\coloneqq f\circ g^n\colon\{0,1\}^{b n} \times \{0,1\}^{b n} \to \{0,1\}$ be the two-party functions from the lifting theorem \cref{thm:lifting}.
We will show that \cref{thm:complement} holds for the language $F^{-1}(1)$.

First we argue that $F$ has an unambiguous DNF of small width.
Indeed, $g$ and $\neg g$ have unambiguous $2 b$-DNFs, which can be extracted from the deterministic decision tree of~$g$.
By plugging these unambiguous $2 b$-DNFs for $g$ and~$\neg g$ into the unambiguous $k$-DNF for~$f$ (and ``multiplying out''), one obtains an unambiguous $2 b k$-DNF, say~$D$, for~$F$.

Over the $2 b n$ variables of~$F$, there exist at most $(2(2 b n) + 1)^{2 b k}$ different conjunctions of at most $2 b k$ literals.
So $D$ consists of at most $n^{O(b k)}$ conjunctions.
From~$D$ we obtain a UFA~$\A$ that recognizes $F^{-1}(1) \subseteq \{0,1\}^{2 b n}$, as follows.
Each initial state of~$\A$ corresponds to a conjunction in~$D$.
When reading the input $x \in \{0,1\}^{2 b n}$, the UFA checks that the corresponding assignment to the variables satisfies the conjunction represented by the initial state.
This check requires at most $O(b n)$ states for each initial state.
Thus, $\A$ has at most $n^{O(b k)} = 2^{\tO(k)} \eqqcolon N$ states in total. (We use $N$ in place of $n$ in the statement of \cref{thm:complement}.)

On the other hand, by \cref{thm:lifting}, we have $\Non_0(F) = \Omega(\C_0(f) \cdot b) = \tOmega(k^2)$.
So by \cref{lem:NFA-CC} any NFA that recognizes $F^{-1}(0)$ has at least $2^{\tOmega(k^2)}$ states.
Any NFA that recognizes $\{0,1\}^* \setminus L(\A)$ can be transformed into an NFA that recognizes $F^{-1}(0) = \{0,1\}^{2 b n} \setminus L(\A)$ by taking a product with a DFA that has $2 b n + 2$ states.
It follows that any NFA that recognizes $\{0,1\}^* \setminus L(\A)$ has at least $2^{\tOmega(k^2)} / (2 b n + 2) = 2^{\tOmega(k^2)}  = N^{\tOmega(\log N)}$ states.
\qed

%% file: parts/union.tex
\section{UFA Union} \label{sec:union}

In this section, we prove \Cref{thm:union}.
\union*

We follow the same high-level approach that we already saw in \cref{sec:complement}. Namely, we will first show that computing the $\lor$-operation is hard for unambiguous DNFs and then lift that hardness to unambiguous protocols, which then implies the same hardness for~UFAs. There are, however, two challenges in carrying out this plan.

\begin{enumerate}[itemsep=1em]
\item It is an open problem to prove an unambiguous lifting theorem. That is, it is not known whether the unambiguous communication complexity of $f\circ g^n$ is at least~$\Omega(\UC_1(f))$. To circumvent this issue, we study instead a \emph{linear relaxation} of unambiguous DNFs. These objects are called \emph{conical juntas} and they do admit a lifting theorem~\cite{GLMWZ16,Kothari2021}.
\item There is no existing result showing that the $\lor$-operation is hard for unambiguous DNFs and/or conical juntas. We show a result of this type. The proof is by a reduction to the hardness of negating conical juntas, which \emph{is} a known result~\cite{GJPW18}.
\end{enumerate}

\subsection{Conical juntas}

A nonnegative function $h\colon \bool^n\to \bR_{\geq 0}$ is a \emph{$d$-junta} if $h$ depends on at most $d$ variables. For example, a conjunction of $d$ literals is a $d$-junta. Moreover, we say $f\colon\bool^n\to \bR_{\geq 0}$ a \emph{conical $d$-junta} if it can be written as a nonnegative linear combination of $d$-juntas. Equivalently,~$f$ is a conical $d$-junta if it can be written as $f=\sum_i w_i C_i$ where each $C_i$ is a width-$d$ conjunction and $w_i\in\bR_{\geq 0}$ are nonnegative coefficients. For example, if $f$ can be written as an unambiguous $d$-DNF, $f= C_1\lor\cdots\lor C_m$, then $f=\sum_i C_i$ is a conical $d$-junta with $0/1$ coefficients. The \emph{nonnegative degree} of $f$, denoted $\deg^+(f)$, is the least $d$ such that~$f$ is a conical $d$-junta. In particular, if $f$ is boolean-valued, then $\deg^+(f)\leq \UC_1(f)$.

We also need to work with \emph{approximate} conical juntas that compute a given function only to within some point-wise error $\epsilon>0$. This is important because the available lifting theorems for $\deg^+$ incur some error, and hence we need to prove lower bounds that are robust to this error. Indeed, we define the \emph{$\epsilon$-approximate nonnegative degree} of $f$, denoted $\deg^+_\epsilon(f)$, as the least nonnegative degree of a conical junta $g$ such that
\[
|f(x)-g(x)| ~\le~ \epsilon\qquad \text{ for all } x\in \bool^n.
\]

\begin{remark*}
An awkward aspect of working with approximate conical juntas is that the error parameter $\epsilon$ is not well behaved. For $0<\epsilon<\delta<1/2$ we of course have $\deg^+_\delta(f)\leq\deg^+_\epsilon(f)$ but it is not a priori clear whether the converse inequality holds with a modest loss in the degree. In fact, in \cref{sec:app}, we will end up showing that there can be a polynomial gap between the nonnegative degrees corresponding to two different error parameters---and this is related to our bonus result discussed in the introduction. As a consequence, our theorems in this section have to track the error parameters with some care.
\end{remark*}

\subparagraph{Linear programming formulation.}
Approximate nonnegative degree can be captured using an LP. Write $\cC_d^n$ for the set of all conjunctions of width at most $d$ over $n$ variables. In the \cref{junta_primal} programme below, we have a variable $w_C\in\bR$ for every $C\in\cC_d^n$. In the associated \cref{junta_dual} programme, we have a variable $\Phi(x)\in\bR$ for each $x\in\bool^n$.
\def\arraystretch{1.3}
\begin{mdframed}
\begin{equation} \label{junta_primal}
\begin{array}{rll}
\min & \epsilon\\
\textsl{subject to}
& \sum_C |w_CC(x)-f(x)| ~\le~ \epsilon,\quad\mbox{} & \forall x\in\bool^n \\
& w_C ~\geq~ 0, & \forall C\in\cC_d^n
\end{array}
\tag{Primal}
\end{equation}
{\color{gray}\noindent\rule{\textwidth}{0.5pt}}%
\begin{equation} \label{junta_dual}
\begin{array}{rll}
\max & \langle \Phi,f\rangle ~\coloneqq~ \sum_x \Phi(x)f(x)\\
\textsl{subject to}
& \|\Phi\| ~\coloneqq~ \sum_x|\Phi(x)| ~\le~ 1 \qquad\mbox{} \\
& \langle \Phi,C\rangle ~\leq~ 0, & \forall C\in\cC_d^n
\end{array}
\tag{Dual}
\end{equation}
\end{mdframed}

We have that $\deg^+_\delta(f)\leq d$ iff the optimal value of \cref{junta_primal} is at most $\delta$. Alternatively, by strong LP duality, we have $\deg^+_\delta(f)> d$ iff there exists a feasible solution $\Phi$ to \cref{junta_dual} such that $\langle \Phi,f\rangle  > \delta$. It is typical to think of such feasible $\Phi\colon\{0,1\}^n\to\bR$ as a \emph{dual certificate} that witnesses a lower bound on approximate nonnegative degree.


\subsection{Hardness of \texorpdfstring{\boldmath $\lor$}{OR}}

The goal of this subsection is to prove \cref{thm:hard-or} below, which states that the $\lor$-operation is hard for unambiguous DNFs and even approximate conical juntas. Given an $n$-bit boolean function $f$ we define a $2n$-bit function by $\ORf(xy) \coloneqq f(x)\lor f(y)$ where $x,y\in\bool^n$. 
\begin{theorem}[\boldmath Hardness of $\lor$] \label{thm:hard-or}
For every $m\in\bN$, there exists a boolean function $f\colon\bool^n\to\bool$ with $n\leq\poly(m)$ such that $\UC_1(f)\leq m$ and $\deg^+_{1.5\times10^{-5}}(f^\lor)\geq\tOmega(m^2)$.
\end{theorem}

We show \cref{thm:hard-or} by combining two lemmas, \cref{lem:gjpw,neg_to_or}, below. The first lemma, proved in~\cite{GJPW18}, states that unambiguous DNFs are hard to negate, even by approximate conical juntas. The second lemma, which remains for us to prove, states that, for conical juntas, computing $f^\lor$ is at least as hard as computing the negation $\neg f$. Hence \cref{thm:hard-or} follows immediately by combining these lemmas.
\begin{lemma}[\boldmath Hardness of $\neg$~{\cite[Lemma 8]{GJPW18}}] \label{lem:gjpw}
For every $m\in\bN$, there exists a boolean function $f\colon\bool^n\to\bool$ with $n\leq\poly(m)$ such that $\UC_1(f)\leq m$ and $\deg^+_{0.05}(\neg f)\geq\tOmega(m^2)$. \qed
\end{lemma}
\begin{lemma}[\boldmath $\lor$ harder than $\neg$]\label{neg_to_or}
For every $\delta>0$ there exists an $\epsilon=\epsilon(\delta)>0$ such that for every boolean function $f$, we have $\andeg{\epsilon}(\ORf)\geq \Omega(\andeg{\delta}(\neg f))$. Moreover, \smash{$\epsilon \coloneqq \big(\frac{\ln(1+\delta)}{\lceil \log_{3/4}\delta\rceil}\big)^2$}.
\end{lemma}

It remains to prove \cref{neg_to_or}. We do it in two steps. In \cref{negp_to_or} we show that the approximate nonnegative degree of $\ORf$ is at least that of $2-f$ by exhibiting a dual certificate. Then in \cref{neg_to_negp} we show that the approximate nonnegative degree of $2-f=1+\neg f$ is at least that of~$\neg f$ via a powering trick. The error parameter $\epsilon$ will degrade in both of these steps. (We will later see that this degradation is, in fact, unavoidable; see \cref{sec:app}.)

\begin{claim}\label{negp_to_or}
We have $\andeg{\epsilon^2}(\ORf)\geq \andeg{\epsilon}(2-f)$ for any boolean-valued $f$ and error $\epsilon$.
\end{claim}
\begin{proof}
Let $d\coloneqq \andeg{\epsilon}(2-f)$ and let $\Psi\colon \bool^n\to \bool$ be a dual certificate witnessing this. That is, $\left<\Phi,2-f\right> > \epsilon$, $\|\Phi\|\le 1$, and $\left<\Phi, C\right> \le 0$ for all $C\in \cC_{d-1}^n$. To construct a dual certificate~$\Psi^\lor\colon \bool^{2n}\to \bool$ witnessing $\andeg{\epsilon^2}(\ORf)\ge d$, we consider the negated tensor product (which was found by an educated guess)
\[
\Phi^\lor(x,y) ~\coloneqq~ -\Phi(x)\Phi(y).
\]
It remains to check that this is feasible for the dual programme and also that $\left<\Phi^\lor,\ORf\right>>\epsilon^2$.
\begin{enumerate}
\item $\left\|\Phi^\lor\right\|=\sum_{x,y} |\Phi^\lor(x,y)|=\sum_{x,y} |\Phi(x)\Phi(y)|=\sum_{x,y} |\Phi(x)|\cdot |\Phi(y)|=\left\|\Phi\right\|^2\le 1$.

\item For any conjunction $C\in \cC_{d-1}^{2n}$, we write $C(x,y)=C_1(x)C_2(y)$ where $C_1,C_2\in\cC_{d-1}^n$. Now
\begin{align*}\textstyle
\left<\Phi^\lor,C\right>~
&\textstyle=~\sum_{x,y} \Phi^\lor(x,y)C(x,y) \\
&\textstyle=~\sum_{x,y} -\Phi(x)\Phi(y)\cdot C_1(x)C_2(y)\\
&\textstyle=~-\big[\sum_x \Phi(x)C_1(x)\big]\big[\sum_y \Phi(y)C_2(y)\big] \\
&\textstyle=~-\left<\Phi,C_1\right>\left<\Phi,C_2\right> \\
&\textstyle\le~ 0.
\end{align*}

\item Observe that $\left<\Phi,-f\right>\geq \left<\Phi,2-f\right>$ since $1\in \cC_{d-1}^n$ for the constant-1 function. Thus
\begin{align*}
\left<\Phi^\lor,\ORf\right>
&\textstyle~=~\sum_{x,y} \Phi^\lor(x,y)(f(x)+f(y)-f(x)f(y))\\
&\textstyle~=~\sum_{x,y} -\Phi(x)\Phi(y)(f(x)+f(y)-f(x)f(y))\\
&\textstyle~=~\sum_{x,y} -\Phi(x)\Phi(y)(2f(x)-f(x)f(y))\\
&\textstyle~=~\sum_{x} -\Phi(x)f(x)\cdot \big[ \sum_{y} \Phi(y)(2-f(y)) \big]\\
&\textstyle~=~\left<\Phi,-f\right>\cdot\left<\Phi,2-f\right>\\
&\textstyle~\geq~\left<\Phi,2-f\right>\cdot\left<\Phi,2-f\right> \tag{above observation}\\
&\textstyle~>~ \epsilon^2. \qedhere
\end{align*}
\end{enumerate}
\end{proof}

\begin{claim}\label{neg_to_negp}
For any $\delta>0$ define \smash{$\epsilon\coloneqq\frac{\ln(1+\delta)}{\lceil \log_{3/4}\delta\rceil}>0$}. Then for any boolean-valued function~$f$ we have  $\andeg{\epsilon}(1+f) \geq \Omega(\andeg{\delta}(f))$.
\end{claim}
\begin{proof}
We may assume $\delta<1/2$ (and hence $\epsilon<1/4$) as otherwise the claim is trivial.
Suppose $\andeg{\epsilon}(1+f)=d$ is witnessed by a conical $d$-junta $g$ that $\epsilon$-approximates $1+f$. Define $g'\coloneqq ((g+\epsilon)/2)^k$ where the exponent is $k\coloneqq \lceil \log_{3/4}\delta\rceil$. By multiplying out the terms in this definition, we see that $g'$ has nonnegative degree $kd=O(d)$. We claim that $g'$ is a $\delta$-approximation of $f$. Indeed, if $f(x)=0$, then $g'(x)\le (1/2+\epsilon)^k\le (3/4)^k\le \delta$. If $f(x)=1$, then $1\le (g(x)+\epsilon)/2 \le 1+\epsilon$, and thus $1\le g'(x)\le (1+\epsilon)^k\le \exp(\epsilon k)\le 1+\delta$.
\end{proof}

\begin{proof}[Proof of \Cref{neg_to_or}]
Using \Cref{negp_to_or} and \Cref{neg_to_negp} (but with $\neg f$ in place of $f$), we have, for any $\delta>0$ and \smash{$\epsilon\coloneqq\frac{\ln(1+\delta)}{\lceil \log_{3/4}\delta\rceil}$}:
\begin{equation*}
\andeg{\epsilon^2}(\ORf)
~\geq~ 
\andeg{\epsilon}(2-f)
~=~
\andeg{\epsilon}(1+\neg f)
~\geq~
\Omega(\andeg{\delta}(\neg f)). \qedhere
\end{equation*}
\end{proof}

\subsection{Unambiguous protocols and nonnegative rank}

Our goal will be to \emph{lift} the hardness of the $\lor$-operation (\cref{thm:hard-or}) to communication complexity. In this subsection, we recall the concepts that are needed for this goal, namely, unambiguous protocols, (approximate) nonnegative rank, and a lifting theorem from nonnegative degree to nonnegative rank~\cite{GLMWZ16,Kothari2021}.

\subparagraph{Unambiguous protocols.} 
Recall from \cref{sub:com-com} the notions of nondeterministic protocols and rectangle covers.
For a two-party function $F\colon X \times Y \to \{0,1\}$, the \emph{partition number} $\Par_1(F)$ is the least number of \emph{pairwise disjoint} rectangles that cover~$F^{-1}(1)$.
Note that $\Cov_1(F) \le \Par_1(F)$.
The \emph{unambiguous communication complexity} of~$F$ is defined as $\Una_1(F) \coloneqq \log_2 \Par_1(F)$.
Note that $\Non_1(F) \le \Una_1(F)$.
Unambiguous communication complexity can be interpreted as the least communication cost of a nondeterministic protocol that has at most one accepting computation on every input. We also have the following folklore lemma, proved the same way as \cref{lem:NFA-CC}, which states that UFAs are simulated by unambiguous protocols.%
\begin{lemma} \label{lem:UFA-CC}
If a two-party function $F\colon \{0,1\}^{m_1} \times \{0,1\}^{m_2} \to \{0,1\}$ admits an UFA with $s$~states, then $\Par_1(F) \le s$ (that is, $\Una_1(F)\leq\log s$). \qed
\end{lemma}

\subparagraph{Nonnegative rank.}
We often think of a two-party function $F\colon X \times Y \to \{0,1\}$ as a boolean matrix $F\in \{0,1\}^{X \times Y}$, sometimes called the \emph{communication matrix} of $F$. For a nonnegative matrix $M\in \bR_{\geq0}^{X \times Y}$ we define its \emph{nonnegative rank}, denoted $\nrank(M)$, as the least $r$ such that~$M$ can be written as a sum of $r$ nonnegative rank-1 matrices, i.e., $M=\sum_{i=1}^r u_iv_i^{\textrm{T}}$, where $u_i\in\bR^X_{\geq0}$ and $v_i\in\bR^Y_{\geq0}$ are nonnegative vectors. Note that for a boolean matrix~$F$,
\begin{equation}\label{eq:una-nrank}
\Par_1(F) ~\geq~ \nrank(F) \qquad \text{and thus}\qquad \Una_1(F) ~\geq~ \log\nrank(F).
\end{equation}
Indeed, if $F^{-1}(1)$ can be partitioned into $r$ rectangles, $F^{-1}(1)= R_1\sqcup\cdots\sqcup R_r$, then $F$ can be written as a sum of $r$ nonnegative rank-1 matrices, $F=M_1+\cdots+M_r$, where $M_i$ is $1$ on the rectangle $R_i$ and $0$ elsewhere. As with nonnegative degree, we define an approximate version of nonnegative rank. 
The \emph{$\epsilon$-approximate nonnegative rank} of $M$, denoted $\anrank{\epsilon}(M)$, is defined as the least $\nrank(N)$ over all nonnegative matrices $N$ that $\epsilon$-approximate $M$, i.e.,
\[
|M_{ij}-N_{ij}| ~\leq~ \epsilon \qquad \text{for all } i,j.
\]

\subparagraph{Nonnegative lifting.}
Finally, we formulate a theorem that lifts lower bounds on the nonnegative degree of an $n$-bit boolean function $f$ to the nonnegative rank of the composed function $F=f\circ g^n$ (which was defined in \cref{sub:com-com}).

\begin{theorem}[{\cite{GLMWZ16,Kothari2021}}] \label{thm:ndeg-lifting}
Fix constants $\delta>\epsilon>0$. For any $n \in \N$ there is a gadget $g\colon \{0,1\}^b \times \{0,1\}^b \to \{0,1\}$ with $b = \Theta(\log n)$ such that for any $f\colon \{0,1\}^n \to \{0,1\}$ we have
\[
\log\anrank{\epsilon}(f\circ g^n) ~\geq~ \Omega(\deg^+_\delta(f)\cdot b). \tag*{\qed}
\]
\end{theorem}

\subsection{Proof of \Cref{thm:union} (and also \cref{thm:or})} \label{sec:proof-2}

We start with the function $f\colon \bool^n\to \bool$ given by \cref{thm:hard-or} such that for $m=\poly(n)$,
\begin{align}
\UC_1(f) &~\leq~ m, \label{eq:upper} \\
\deg^+_{1.5\times 10^{-5}}(f^\lor) &~\geq~ \tOmega(m^2). \label{eq:lower}
\end{align}
We then use the gadget $g$ on $b=\Theta(\log n)$ bits from the lifting theorem \cref{thm:ndeg-lifting} to construct $F\coloneqq f\circ g^n$. By the same argument as in \cref{sec:proof-1} we see that the resulting $F\colon \bool^{nb}\times \bool^{nb}\to\bool$ enjoys the following upper bounds, derived from \cref{eq:upper}.
\begin{itemize}
\item $F$ admits an unambiguous DNF of width $2bm=\tO(m)$.
\item $F$ admits an UFA of size $2^{\tO(m)}$.
\item $F$ admits an unambiguous protocol of cost $\Una_1(F)\leq \tO(m)$.
\end{itemize}
On the other hand, we note that $F^\lor = (f\circ g^n)^\lor = f^\lor \circ g^n$. Hence, we may combine \cref{eq:una-nrank}, \cref{thm:ndeg-lifting}, and \cref{eq:lower} to conclude that
\begin{equation}\label{eq:nrank-lower}
\Una_1(F^\lor)
~\geq~ \log \nrank_{10^{-5}}(F^\lor)
~\geq~ \Omega(\ndeg_{1.5\times10^{-5}}(f^\lor))
~\geq~ \tOmega(m^2).
\end{equation}
This finishes the proof of \cref{thm:or}. We proceed with the proof of \cref{thm:union}. To this end, we define two languages 
\begin{align*}
 L_1 &~\coloneqq~ \{xx'yy' : x,x',y,y'\in \bool^{bn} \text{ and } F(x,y)=1\}, \\
 L_2 &~\coloneqq~ \{xx'yy' : x,x',y,y'\in \bool^{bn} \text{ and } F(x',y')=1\}.
\end{align*}
Both $L_1$ and $L_2$ admit UFAs of size $\poly(n)\cdot 2^{\tO(m)}=2^{\tO(m)}\eqqcolon N$. By contrast, we have~$L_1\cup L_2 = (F^\lor)^{-1}(1)$, and this union language requires UFAs of size $2^{\tOmega(m^2)}=N^{\tOmega(\log N)}$ by~\cref{eq:nrank-lower} and \cref{lem:UFA-CC}. This concludes the proof of \cref{thm:union}. \qed

%% file: parts/separation.tex
\section{UFA Separation}

In this section, we prove \Cref{thm:separation}.
\separation*

Loosely speaking, in our construction, we define NFAs $\A_1, \A_2$ that recognize \emph{(sparse) set disjointness} and its complement.
For $n \in \N$ and $k \le n$ we define
\begin{align*}
\Disj^n_k\ &\coloneqq\ \{(S,T) \mid S \subseteq [n],\ T \subseteq [n],\ |S| = |T| = k,\ S \cap T = \emptyset\}\,.
\end{align*}
Define also
 $\enc{\Disj^n_k} \coloneqq \{\enc{S} \enc{T} \mid (S,T) \in \Disj^n_k\}$ where $\enc{S} \in  \{0,1\}^n$ is such that the $i$th letter of $\enc{S}$ is~$1$ if and only if $i \in S$, and similarly for $\enc{T}$.
Note that $\enc{S}, \enc{T}$ each contain $k$ times the letter~$1$.
To prove~\cref{thm:separation} it suffices to prove the following lemma.
\begin{lemma} \label{lem:separation}
For any $n \in \N$ let $k \coloneqq \lceil \log_2 n \rceil$.
There are NFAs $\A_1, \A_2$ with $n^{O(1)}$ states such that $L(\A_1) = \enc{\Disj^n_k}$ and $L(\A_2) = \bool^* \setminus \enc{\Disj^n_k}$.
Any UFA that recognizes $\enc{\Disj^n_k}$ has at least $n^{\Omega(\log n)}$ states.
\end{lemma}
In the rest of the section we prove \cref{lem:separation} by following Razborov's analysis of sparse set disjointness~\cite{Razborov1990}. In particular, we will give a self-contained proof of the existence of polynomial-sized NFAs for $\enc{\Disj^n_k}$ and its complement, but the main argument also comes from communication complexity.

\subsection{Proof of \texorpdfstring{\cref{lem:separation}}{Lemma~\ref{lem:separation}}}

First we prove the statement on UFAs.
Write $\binom{[n]}{k} \coloneqq \{S \subseteq [n] \mid |S| = k\}$.
Let $F \colon \binom{[n]}{k} \times \binom{[n]}{k} \to \{0,1\}$ be the two-party function with $F(S,T) = 1$ if and only if $(S,T) \in \Disj^n_k$.
It is shown, e.g., in \cite[Example~2.12]{KushilevitzNisan} 
 that the communication matrix of $F$ has full rank, $\rank(F)=\binom{n}{k}$.
Let $F'\colon \{0,1\}^n \times \{0,1\}^n \to \{0,1\}$ be such that $F'(x,y) = 1$ if and only if $x y \in \enc{\Disj^n_k}$.
Then $F$ is a principal submatrix of~$F'$, so $\binom{n}{k} \le \rank(F')$.
Using~\cref{eq:una-nrank,lem:UFA-CC} it follows that any UFA, say~$\A$, that recognizes $\enc{\Disj^n_k}$ has at least $\binom{n}{k} \ge (\frac{n}{k})^k$ states.
With $k \coloneqq \lceil \log_2 n \rceil$, it follows that $\A$ has $n^{\Omega(\log n)}$ states.

It is easy to see that there is an NFA, $\A_2$, with $n^{O(1)}$~states and $L(\A_2) = \bool^* \setminus \enc{\Disj^n_k}$.
Indeed, we can assume that the input is of the form $\enc{S} \enc{T}$; otherwise $\A_2$ accepts.
NFA~$\A_2$ guesses $i \in [n]$ such that $i \in S \cap T$ and then checks it.

Finally, we show that there is an NFA, $\A_1$, with $n^{O(1)}$~states and $L(\A_1) = \enc{\Disj^n_k}$.
We can assume that the input is of the form $\enc{S} \enc{T}$; otherwise $\A_1$ rejects.
NFA~$\A_1$ ``hard-codes'' polynomially many sets $Z_1, \ldots, Z_\ell \subseteq [n]$.
It guesses $i \in [\ell]$ such that $S \subseteq Z_i$ and $Z_i \cap T = \emptyset$ and then checks it.
It remains to show that there exist $\ell = n^{O(1)}$ sets $Z_1, \ldots, Z_\ell \subseteq [n]$ such that for any $(S,T) \in \Disj^n_k$ there is $i \in [\ell]$ with $S \subseteq Z_i$ and $Z_i \cap T = \emptyset$.
The argument uses the probabilistic method and is due to~\cite{Razborov1990}; see also \cite[Example~2.12]{KushilevitzNisan}.
We reproduce it here due to its elegance and brevity.

Fix $(S,T) \in \Disj^n_k$.
Say that a set $Z \subseteq [n]$ \emph{separates}~$(S,T)$ if $S \subseteq Z$ and $Z \cap T = \emptyset$.
A random set $Z \subseteq [n]$ (each $i$ is in~$Z$ with probability~$1/2$) separates~$(S,T)$ with probability~$2^{-2 k}$.
Thus, choosing $\ell \coloneqq \big\lceil 2^{2 k} \ln \binom{n}{k}^2 \big\rceil = n^{O(1)}$ random sets~$Z \subseteq [n]$ independently, the probability that none of them separates~$(S,T)$ is
\[\textstyle
 (1 - 2^{-2 k})^\ell  ~<~ e^{-2^{-2 k} \ell}  ~\le~ \binom{n}{k}^{-2}\,.
\]
By the union bound, since $|\Disj^n_k| < \binom{n}{k}^2$, the probability that there exists $(S,T) \in \Disj^n_k$ such that none of $\ell$ random sets separates~$(S,T)$ is less than~$1$.
Equivalently, the probability that for all $(S,T) \in \Disj^n_k$ at least one of $\ell$ random sets separates~$(S,T)$ is positive.
It follows that there are $Z_1, \ldots, Z_\ell \subseteq [n]$ such that each~$(S,T) \in \Disj^n_k$ is separated by some~$Z_i$. 
\qed


%% file: parts/applications.tex
\section{Bonus result: Approximate nonnegative rank} \label{sec:app}

In this section, we prove \cref{thm:error}.

\error*

We first illustrate the idea in the context of nonnegative degree. In contrast to \cref{thm:hard-or} (which states that $\lor$ is hard to approximate to within tiny error), we show that the $\lor$-operation is, in fact, easy to approximate when we allow large enough error.
\begin{claim} \label{cl:or}
For any boolean-valued $f$, we have $\andeg{1/4}(\ORf)\le \ndeg(f)$.
\end{claim}
\begin{proof}
Let $g\colon \bool^{2n}\rightarrow \bR_{\geq 0}$ be given by $g(x,y)\coloneqq (f(x)+f(y))/2+1/4$. Then 
\begin{equation*}
    g(x,y) ~=~
    \begin{cases}
    1/4 & \text{if }f(x)=f(y)=0, \\
    5/4 & \text{if }f(x)=f(y)=1, \\
    3/4 & \text{otherwise}.
    \end{cases}
\end{equation*}
Thus $g$ is a $1/4$-approximation to $\ORf$. Note also that $\ndeg(g)\le \ndeg(f)$, as desired.
\end{proof}

We can now repeat the same idea for nonnegative rank. In \cref{sec:proof-2} we constructed a boolean matrix (two-party function) $F$ such that $\log\nrank(F)\leq \Una_1(F) \leq m$ and $\log\anrank{10^{-5}}(F^\lor)\geq \tOmega(m^2)$. We claim that $\log\anrank{1/4}(F^\lor)\leq O(m)$, which would finish the proof of \cref{thm:error}. Indeed, analogously to \cref{cl:or}, we can define a nonnegative matrix by $G(xx,yy')\coloneqq (F(x,y)+F(x',y'))/2+1/4$. This is a $1/4$-approximation to $F^\lor$ and we have $\nrank(G)\leq 2\cdot\nrank(F)+1\leq 2^{m+1}+1$, as claimed.

%% file: ufa_main.bbl
\begin{thebibliography}{10}

\bibitem{Alon2003}
Noga Alon.
\newblock Problems and results in extremal combinatorics--{I}.
\newblock {\em Discrete Mathematics}, 273(1--3):31--53, 2003.
\newblock \href {https://doi.org/10.1016/S0012-365X(03)00227-9}
  {\path{doi:10.1016/S0012-365X(03)00227-9}}.

\bibitem{Balodis2021FOCS}
Kaspars Balodis, Shalev Ben-David, Mika G\"o\"os, Siddhartha Jain, and Robin
  Kothari.
\newblock Unambiguous {DNF}s and {A}lon-{S}aks-{S}eymour.
\newblock In {\em 62nd {IEEE} Annual Symposium on Foundations of Computer
  Science (FOCS)}, 2021.
\newblock To appear. Available at \url{https://arxiv.org/abs/2102.08348}.

\bibitem{Birget1993}
Jean-Camille Birget.
\newblock Partial orders on words, minimal elements of regular languages, and
  state complexity.
\newblock {\em Theoretical Computer Science}, 119(2):267--291, 1993.
\newblock \href {https://doi.org/10.1016/0304-3975(93)90160-U}
  {\path{doi:10.1016/0304-3975(93)90160-U}}.

\bibitem{automata-toolbox}
Miko{\l}aj Bojańczyk and Wojciech Czerwiński.
\newblock {\em An Automata Toolbox}.
\newblock 2018.
\newblock Available at
  \url{https://www.mimuw.edu.pl/~bojan/paper/automata-toolbox-book}.

\bibitem{Colcombet15}
Thomas Colcombet.
\newblock Unambiguity in automata theory.
\newblock In {\em 17th International Workshop on Descriptional Complexity of
  Formal Systems (DCFS)}, volume 9118 of {\em Lecture Notes in Computer
  Science}, pages 3--18. Springer, 2015.
\newblock \href {https://doi.org/10.1007/978-3-319-19225-3\_1}
  {\path{doi:10.1007/978-3-319-19225-3\_1}}.

\bibitem{Colcombet2021}
Thomas Colcombet, Karin Quaas, and Micha{\l} Skrzypczak.
\newblock Unambiguity in automata theory ({D}agstuhl seminar 21452).
\newblock {\em Dagstuhl Reports}, 2021.
\newblock To appear.

\bibitem{CzerwinskiD19}
Wojciech Czerwiński, Laure Daviaud, Nathana{\"{e}}l Fijalkow, Marcin
  Jurdziński, Ranko Lazić, and Pawel Parys.
\newblock Universal trees grow inside separating automata: Quasi-polynomial
  lower bounds for parity games.
\newblock In {\em Proceedings of the 2019 Annual ACM-SIAM Symposium on Discrete
  Algorithms (SODA)}, pages 2333--2349, 2019.
\newblock \href {https://doi.org/10.1137/1.9781611975482.142}
  {\path{doi:10.1137/1.9781611975482.142}}.

\bibitem{CzerwinskiL19}
Wojciech Czerwiński and Sławomir Lasota.
\newblock Regular separability of one counter automata.
\newblock {\em Logical Methods in Computer Science}, 15(2), 2019.
\newblock \href {https://doi.org/10.23638/LMCS-15(2:20)2019}
  {\path{doi:10.23638/LMCS-15(2:20)2019}}.

\bibitem{GaoMRY16}
Yuan Gao, Nelma Moreira, Rog{\'{e}}rio Reis, and Sheng Yu.
\newblock A survey on operational state complexity.
\newblock {\em Journal of Automata, Languages and Combinatorics},
  21(4):251--310, 2016.
\newblock \href {https://doi.org/10.25596/jalc-2016-251}
  {\path{doi:10.25596/jalc-2016-251}}.

\bibitem{Goos15}
Mika G{\"{o}}{\"{o}}s.
\newblock Lower bounds for clique vs.\ independent set.
\newblock In {\em {IEEE} 56th Annual Symposium on Foundations of Computer
  Science {(FOCS)}}, pages 1066--1076. {IEEE} Computer Society, 2015.
\newblock \href {https://doi.org/10.1109/FOCS.2015.69}
  {\path{doi:10.1109/FOCS.2015.69}}.

\bibitem{GJPW18}
Mika G{\"o}{\"o}s, T.S. Jayram, Toniann Pitassi, and Thomas Watson.
\newblock Randomized communication versus partition number.
\newblock {\em ACM Transactions on Computation Theory (TOCT)}, 10(1):1--20,
  2018.
\newblock \href {https://doi.org/10.1145/3170711} {\path{doi:10.1145/3170711}}.

\bibitem{GLMWZ16}
Mika G{\"{o}}{\"{o}}s, Shachar Lovett, Raghu Meka, Thomas Watson, and David
  Zuckerman.
\newblock Rectangles are nonnegative juntas.
\newblock {\em SIAM Journal on Computing}, 45(5):1835--1869, 2016.
\newblock \href {https://doi.org/10.1137/15M103145X}
  {\path{doi:10.1137/15M103145X}}.

\bibitem{IK21}
Emil Indzhev and Stefan Kiefer.
\newblock On complementing unambiguous automata and graphs with many cliques
  and cocliques.
\newblock Technical report, arxiv.org, 2021.
\newblock Available at \url{https://arxiv.org/abs/2105.07470}.

\bibitem{Jiraskova05}
Galina Jir{\'{a}}skov{\'{a}}.
\newblock State complexity of some operations on binary regular languages.
\newblock {\em Theoretical Computer Science}, 330(2):287--298, 2005.
\newblock \href {https://doi.org/10.1016/j.tcs.2004.04.011}
  {\path{doi:10.1016/j.tcs.2004.04.011}}.

\bibitem{JirasekJS18}
Jozef~Jir{\'{a}}sek Jr., Galina Jir{\'{a}}skov{\'{a}}, and Juraj \v{S}ebej.
\newblock Operations on unambiguous finite automata.
\newblock {\em International Journal of Foundations of Computer Science},
  29(5):861--876, 2018.
\newblock \href {https://doi.org/10.1142/S012905411842008X}
  {\path{doi:10.1142/S012905411842008X}}.

\bibitem{KMSY14}
Gillat Kol, Shay Moran, Amir Shpilka, and Amir Yehudayoff.
\newblock Approximate nonnegative rank is equivalent to the smooth rectangle
  bound.
\newblock {\em Computational Complexity}, 28(1):1--25, 2019.
\newblock Preliminary version in \emph{ICALP~'14}.
\newblock \href {https://doi.org/10.1007/s00037-018-0176-4}
  {\path{doi:10.1007/s00037-018-0176-4}}.

\bibitem{Kothari2021}
Pravesh Kothari, Raghu Meka, and Prasad Raghavendra.
\newblock Approximating rectangles by juntas and weakly exponential lower
  bounds for {LP} relaxations of {CSPs}.
\newblock {\em {SIAM} Journal on Computing}, pages STOC17--305--STOC17--332,
  May 2021.
\newblock Preliminary version in \emph{STOC~'17}.
\newblock \href {https://doi.org/10.1137/17m1152966}
  {\path{doi:10.1137/17m1152966}}.

\bibitem{KushilevitzNisan}
Eyal Kushilevitz and Noam Nisan.
\newblock {\em Communication Complexity}.
\newblock Cambridge University Press, 1997.
\newblock \href {https://doi.org/10.1017/CBO9780511574948}
  {\path{doi:10.1017/CBO9780511574948}}.

\bibitem{Rao2020}
Anup Rao and Amir Yehudayoff.
\newblock {\em Communication Complexity: And Applications}.
\newblock Cambridge University Press, 2020.
\newblock \href {https://doi.org/10.1017/9781108671644}
  {\path{doi:10.1017/9781108671644}}.

\bibitem{Raskin18}
Mikhail Raskin.
\newblock A superpolynomial lower bound for the size of non-deterministic
  complement of an unambiguous automaton.
\newblock In {\em 45th International Colloquium on Automata, Languages, and
  Programming (ICALP 2018)}, volume 107 of {\em Leibniz International
  Proceedings in Informatics (LIPIcs)}, pages 138:1--138:11, 2018.
\newblock \href {https://doi.org/10.4230/LIPIcs.ICALP.2018.138}
  {\path{doi:10.4230/LIPIcs.ICALP.2018.138}}.

\bibitem{Razborov1990}
Alexander Razborov.
\newblock Applications of matrix methods to the theory of lower bounds in
  computational complexity.
\newblock {\em Combinatorica}, 10(1):81--93, 1990.
\newblock \href {https://doi.org/10.1007/BF02122698}
  {\path{doi:10.1007/BF02122698}}.

\bibitem{Wiki}
Wikipedia.
\newblock State complexity.
\newblock URL: \url{https://en.wikipedia.org/wiki/State_complexity}.

\bibitem{Yannakakis1991}
Mihalis Yannakakis.
\newblock Expressing combinatorial optimization problems by linear programs.
\newblock {\em Journal of Computer and System Sciences}, 43(3):441--466, 1991.
\newblock \href {https://doi.org/10.1016/0022-0000(91)90024-Y}
  {\path{doi:10.1016/0022-0000(91)90024-Y}}.

\end{thebibliography}
